\documentclass[onecolumn]{IEEEtran}
\usepackage{multirow}
\usepackage{graphicx}%
\usepackage{multirow}%
\usepackage{amsmath,amssymb,amsfonts}%
\usepackage{amsthm}%
\usepackage{mathrsfs}%
\usepackage{xcolor}%
\usepackage{textcomp}%
\usepackage{manyfoot}%
\usepackage{booktabs}%
\usepackage{algorithm}%
\usepackage{algorithmicx}%
\usepackage{algpseudocode}%
\usepackage{listings}%
\usepackage{cite}

\usepackage{overpic}
\usepackage[h]{esvect}
\usepackage{diagbox}

\newtheorem{theorem}{Theorem}
\newtheorem{corollary}[theorem]{Corollary}%
\newtheorem{lemma}[theorem]{Lemma}%

\newtheorem{example}{Example}%
\newtheorem{definition}{Definition}%

\usepackage{mathtools}
\DeclarePairedDelimiter\abs{\lvert}{\rvert}

\DeclarePairedDelimiter\parenv{\lparen}{\rparen}

\DeclarePairedDelimiter\set{\{}{\}}

\renewcommand{\leq}{\leqslant}

\renewcommand{\geq}{\geqslant}

\newcommand{\cA}{\mathcal{A}}

\newcommand{\cC}{\mathcal{C}}

\newcommand{\cE}{\mathcal{E}}

\newcommand{\cI}{\mathcal{I}}

\newcommand{\cL}{\mathcal{L}}

\newcommand{\N}{\mathbb{N}}
\newcommand{\R}{\mathbb{R}}
\newcommand{\Z}{\mathbb{Z}}

\newcommand{\eqdef}{\triangleq}

\newcommand{\bv}{\boldsymbol{v}}

\newcommand{\bx}{\boldsymbol{x}}
\newcommand{\by}{\boldsymbol{y}}

\DeclareMathOperator{\ccap}{\mathsf{cap}}

\newcommand{\bH}{\mathbf{H}}

\begin{document}


\title{On the Capacity of Sequences of Coloring Channels}

\author{
Wenjun Yu~\IEEEmembership{Member,~IEEE} and Moshe~Schwartz,~\IEEEmembership{Fellow,~IEEE}%
\thanks{Wenjun Yu is with the School
   of Electrical and Computer Engineering, Ben-Gurion University of the Negev,
   Beer Sheva 8410501, Israel
   (e-mail: wenjun@post.bgu.ac.il).}%
\thanks{Moshe Schwartz is with the Department of Electrical and Computer Engineering, McMaster University, Hamilton, ON, L8S 4K1, Canada, and on a leave of absence from the School
   of Electrical and Computer Engineering, Ben-Gurion University of the Negev,
   Beer Sheva 8410501, Israel
   (e-mail: schwartz.moshe@mcmaster.ca).}
}

\maketitle

\begin{abstract}
A single coloring channel is defined by a subset of letters it allows to pass through, while deleting all others. A sequence of coloring channels provides multiple views of the same transmitted letter sequence, forming a type of sequence-reconstruction problem useful for protein identification and information storage at the molecular level. We provide exact capacities of several sequences of coloring channels: uniform sunflowers, two arbitrary intersecting sets, and paths. We also show how this capacity depends solely on a related graph we define, called the pairs graph. Using this equivalence, we prove lower and upper bounds on the capacity, and a tailored bound for a coloring-channel sequence forming a cycle. In particular, for an alphabet of size $4$, these results give the exact capacity of all coloring-channel sequences except for a cycle of length $4$, for which we only provide bounds.
\end{abstract}

\begin{IEEEkeywords}
Coloring channel, capacity, sequence reconstruction
\end{IEEEkeywords}

\section{Introduction}

\IEEEPARstart{T}{he} coloring channels were introduced by~\cite{BarWacYaa25}, motivated by protein-identification applications, and in particular, a method for reading amino-acid sequences that was suggested in~\cite{OhaGirNasSheMel19}. In it, fluorescence markers are attached to the amino acids in a way that allows a nanopore reading to observe them. The general model for this process, suggested in~\cite{BarWacYaa25}, is called the \emph{coloring channel}.

A single coloring channel is defined by a subset of letters from a larger ambient alphabet. When a sequence of letters passes through the channel, only those within the subset get ``colored'' and therefore observed at the channel output. The unobserved letters are hence deleted in the channel output. As an example, if the word ``catapult'' is passed through a coloring channel defined by the letters a, c, and t. The output of the channel is then ``catat'', with the letters l, p, and u, deleted. When ``catapult'' is passed through another coloring channel, we get a different view of the sequence. If this time the coloring channel is defined by the letters c, l, p, and u, we get the output ``cpul''. Given these two outputs, ``catat'' and ``cpul'', we may guess ``catapult'' was transmitted, but another possible transmission might be ``capultat''. Thus, given a sequence of coloring channels, the \emph{coloring-channel problem} requires us to design a code, each of whose codewords are uniquely decodable after passing in parallel through the given channels.

The problem of reconstructing sequences that were passed through multiple deletion channels has a long history, going back to the seminal work of Levenshtein on reconstruction schemes~\cite{Lev01,Lev01a}. In this scheme, a codeword is transmitted in parallel over identical channels, distinct outputs are collected, and a unique decoding is expected. The main goals are finding the minimal number of distinct outputs required for unique decoding given the code used by the transmitter, as well as designing a reconstruction algorithm. These were studied in a sequence of papers~\cite{GabYaa18,AbuYaa21,CaiKiaNguYaa22,ChrKiaYaa22}, recently culminating in~\cite{PhaGoyKia25}, which proved a complete asymptotic solution.

Several crucial differences between the coloring-channel problem and Levenshtein's reconstruction for deletion channels prohibit the use of solutions to the latter when studying the former. First, in Levenshtein's reconstruction all the channels are identical, whereas for the coloring-channel problem, each channel has a different subset defining it. Second, in Levenshtein's reconstruction an adversary operates on every channel. However, in the coloring-channel problem the deletion is determined solely by the subset associated with each channel, and is known in advance. Finally, in Levenshtein's reconstruction the adversary is limited in the sense that a maximal number of deletions is allowed in each channel. In contrast, in the coloring-channel problem the number of deleted symbols is not bounded.

Among the many challenges we may associate with the coloring-channel problem, an important one, and the first stated in~\cite{BarWacYaa25}, is determining the capacity, namely, the asymptotic rate of optimal codes for the given channels. Several capacity results were proved in~\cite{BarWacYaa25}: the capacity of a single coloring channel was determined, as well as the capacity of equal-size pairwise-disjoint channels. A more elaborate case was also addressed, in which two coloring channels are defined by subsets of size $q-1$ each, with an intersection of size $q-2$. It was also proved in~\cite{BarWacYaa25} that a sequence of coloring channels has full capacity if and only if every pair of letters appears together in at least one channel, thus forming a certain covering design.

The main contributions of this paper are as follows: First, we extend the repertoire of coloring-channel sequences for which we know the exact capacity. By viewing these channel sequences as set systems, we can find the exact capacity of uniform sunflowers, two arbitrary intersecting sets, and paths. The capacity is found by solving certain optimization problems, the most difficult of which requires continued fractions and Chebyshev polynomials. These exact capacities generalize the cases for which exact capacities were found in~\cite{BarWacYaa25}. Additionally, we show how certain cases we call \emph{separable} may be reduced to the problem of finding the capacity of a smaller sequence of coloring channels, generalizing the capacity of pairwise-disjoint equal-sized channels proved in~\cite{BarWacYaa25}.

Our second contribution is proving that the capacity of a sequence of coloring channels depends on a graph we call the \emph{pairs graph}. This first implies that, over an alphabet of size $3$, there are essentially only \emph{two} cases of sequences of coloring channels that use all letters, and their capacity may already be obtained through the results of~\cite{BarWacYaa25} (see Table~\ref{tab:q3}). However, already for alphabets of size $4$ the results of~\cite{BarWacYaa25} are insufficient since they only cover two of the possible six cases. Continuing with our contributions, using monotonicity, through the pairs-graph approach we obtain lower and upper bounds on the capacity of general coloring-channel sequences. Additionally, the number of coloring channels may be reduced to the intersection number of the pairs graph, without changing the capacity.

Our final contribution provides specific stronger bounds on the capacity of coloring-channel sequences that form cycles. Using all of these results, for an alphabet of size $4$, we are able to provide exact capacity for all coloring-channel sequences except the cycle of length $4$, for which we only have bounds. These are shown in Table~\ref{tab:q4}.

The paper is organized as follows. We begin in Section~\ref{sec:prelim} by providing notation and definitions used throughout the paper. In Section~\ref{sec:exact} we prove all exact capacity results. In Section~\ref{sec:bound} we define the pairs graph, prove general bounds on the capacity, as well as a specific bound on the capacity of a cycle. We conclude in Section~\ref{sec:conc} with a summary of the results, and a short list of open problems.

\section{Preliminaries}
\label{sec:prelim}

For any $i,j\in\Z$, $i\leq j$, we define $[i,j]\eqdef \set{i,i+1,\dots,j}$. For $n\in\N$ we then define $[n]\eqdef [1,n]$. Consider some finite alphabet $\Sigma$. Without loss of generality, we shall assume throughout the paper that $\Sigma=[q]$ for some integer $q\geq 2$. A sequence (vector) of length $n$ over $\Sigma$ is an $n$-tuple $\bx=x_1x_2\dots x_n$, where $x_i\in\Sigma$ for all $i$. We use $\varepsilon$ to denote the unique sequence of length $0$.

Given a set $A\subseteq [q]$, we use $2^{A}$ for its power set, i.e., $2^A\eqdef \set{B: B\subseteq A}$. We denote $\binom{A}{\ell}\eqdef \set{ B: B\subseteq A, \abs{B}=\ell}$. We also use $A^n\eqdef\set{(a_1,\dots,a_n):a_i\in A,i\in [n]}$, as well as $\overline{A}\eqdef [q]\setminus A$ for the complement. A (finite, undirected) graph $G$ is defined by a pair $G=(V,E)$, where $V$ is some finite set of vertices, and $E\subseteq \binom{V}{2}$ is the set of edges. 

We shall often encounter the binary entropy function, $H(x):(0,1)\to(0,1)$, which is defined as
\[H(x) \eqdef -x\log_2x -(1-x)\log_2(1-x).\]
We extend the function to include $H(0)=H(1)=0$, making it continuous on $[0,1]$. We shall also require the well-known approximation of the binomial coefficient~\cite[p.~309, Lemma 7]{MacSlo78},
\begin{equation}
\label{eq:binom}
\binom{n}{\alpha n (1+o(1)) } = 2^{nH(\alpha)(1+o(1))},
\end{equation}
for all real $\alpha\in[0,1]$, and where $o(1)$ denotes a vanishing function as $n\to\infty$.

We now describe the main channels that we study, the \emph{coloring channels}, following~\cite{BarWacYaa25}. Let $\Sigma=[q]$ be an alphabet of size $q$, and let $\bx = x_1 \dots x_n \in \Sigma^n$ denote a transmitted sequence. We identify a coloring channel $I$ with a non-empty subset of $\Sigma$, i.e., $I\subseteq \Sigma$, $I\neq \emptyset$. Such a noisy channel, $I$, acts symbol-wise via the error function $\cE_I : \Sigma \to \Sigma\cup\set{\varepsilon}$ defined as
\begin{equation}
\label{eq:del}
\cE_I(a) = \begin{cases}
    a & a\in I,\\
    \varepsilon & a\not\in I.
\end{cases}
\end{equation}
For a transmitted sequence $\bx$, the received sequence through channel $I$ is
\[
\bx_I \eqdef \cE_I(x_1) \cE_I(x_2) \dots \cE_I(x_n).
\]
In other words, when a sequence passes through the channel $I$, only letters in $I$ reach the receiver (in the order they were transmitted), whereas all letters in $\overline{I}$ get deleted.

As in~\cite{BarWacYaa25,HorYaa18,JunLaiLeh25}, we generalize our setting to the case where a sequence is transmitted over several channels. Unlike Levenshtein reconstruction~\cite{Lev01,Lev01a}, the channels are not necessarily identical. In the scenario of multiple coloring channels, $\cI = (I_1, \ldots, I_t)$, for any input sequence $\bx \in \Sigma^n$, the received outputs across all channels can be represented as a tuple of sequences
\[\bx_\cI \eqdef (\bx_{I_1},\bx_{I_2},\dots,\bx_{I_t}).\]
We emphasize that the ordering of the channels is arbitrary and needed for channel-identification purposes only. Thus, by abuse of terminology, we shall sometimes refer to $\cI$ as a \emph{set system}, when the order of the channels is immaterial.

\begin{example}
    Let $q=3$, so $\Sigma=\set{1,2,3}$. Let us take $\bx = 3122123\in \Sigma^7$, $I_1 = \set{1,3}$, $I_2=\set{1,2}$, and $\cI=(I_1,I_2)$. Under the coloring channel of~\eqref{eq:del}, we have
    \begin{align*}
    \bx_{I_1}&=3113, & \bx_{I_2}&=12212,
    \end{align*}
    so over $\cI$ we receive,
    \[
    \bx_\cI = (3113,12212).
    \]
\end{example}

We say that two distinct sequences, $\bx,\by\in\Sigma^n$ are \emph{confusable} if $\bx_\cI= \by_\cI$. Intuitively, if that is the case, and $\bx_\cI=\by_\cI$ is received, then the receiver has no way of knowing whether $\bx$ was transmitted and corrupted by the channels into becoming the observed sequences, or whether it was $\by$. A \emph{reconstruction code} for the channels $\cI$ is a subset $\cC \subseteq \Sigma^n$ that does not contain any pair of confusable sequences, namely, $\bx_\cI \neq \by_\cI$ for all distinct $\bx,\by \in \cC$. This condition ensures that every possible received sequence can be uniquely traced back to its transmitted codeword, achieving error-free decoding.

We conveniently denote the set of all possible channel outputs as
\[
\cA_{\cI}(n) \eqdef \set*{\bx_{\cI} : \bx \in \Sigma^n}.
\]
If $n$ is understood from context, we simply write $\cA_{\cI}$. Trivially, confusability (with respect to $\cI$) is an equivalence relation on $\Sigma^n$, and the number of equivalence classes is simply $\abs{\cA_{\cI}}$. Thus, an optimal reconstruction code $\cC$ for channels $\cI$ has exactly one codeword from each equivalence class, and therefore $\abs{\cC}=\abs{\cA_{\cI}}$. The main goal of this paper is to find the asymptotic rate of optimal reconstruction codes for $\cI$, called the \emph{capacity of $\cI$}, and defined by
\[
\ccap(\cI) \eqdef \limsup_{n\to\infty} \frac{1}{n} \max_{\cC} \log_q\abs{\cC}=\limsup_{n\to\infty} \frac{1}{n}\log_q \abs{\cA_{\cI}},
\]
where the maximization is over all reconstruction codes $\cC$ for the sequence of coloring channels $\cI$.

\section{Exact Capacity}
\label{sec:exact}

In this section, we provide the exact capacity of sequences of coloring channels for various useful set systems. In particular, we find the exact capacity of uniform sunflowers, two intersecting sets, and paths.

We begin with some simple observations, with the goal of avoiding trivial cases. Consider the alphabet $\Sigma=[q]$, and a sequence of coloring channels $\cI=(I_1,I_2,\dots,I_t)$, where $I_i\subseteq \Sigma$ for all $i\in[t]$. We first observe that if for some $i\neq j$ we have $I_i\subseteq I_j$, then we can trivially discard $I_i$ completely without affecting the capacity. This follows since for any sequence $\bx\in\Sigma^n$, $\bx_{I_j}$ completely determines $\bx_{I_i}$. Thus, we may assume that no coloring channel is contained in another.

Our next observation is slightly more elaborate. We require the following definition.

\begin{definition}
Let $\cI=(I_1,\dots,I_t)$ be a sequence of coloring channels over $\Sigma$. We say $\cI$ is \emph{separable} if there exists a proper non-empty subset, $\emptyset\subset S \subset [t]$ such that
\[
\bigcup_{i\in S}I_i \cap \bigcup_{i\in\overline{S}} I_i = \emptyset.
\]
\end{definition}

We claim that when a sequence of coloring channels is separable, then the problem of determining the capacity is reduced to a smaller sequence of coloring channels.

\begin{lemma}
\label{lem:trivial}
Let $\cI=(I_1,\dots,I_t)$ be a sequence of coloring channels over $\Sigma=[q]$. Assume $\cI$ is separable with subset $\emptyset\subset S\subset [t]$, and define $\cI_1 =(I_i)_{i\in S}$, $\cI_2 = (I_i)_{i\in \overline{S}}$. Let $c_1=\ccap(\cI_1)$ and $c_2=\ccap(\cI_2)$. Then
\[
\ccap(\cI)= \max \set*{ c_1, c_2 }.
\]
\end{lemma}
\begin{IEEEproof}
Assume w.l.o.g. that $c_1\geq c_2$, and that $S=\set{1,\dots,s}$, $\overline{S}=\set{s+1,\dots,t}$. Let us further denote $T_i(n)\eqdef \abs{\cA_{\cI_i}(n)}$ for $i=1,2$. It now follows that
\begin{equation}
\label{eq:prod}
\abs*{\cA_{\cI}(n)}=\sum_{i=0}^{n} T_1(i)T_2(n-i).
\end{equation}
Let $\delta>0$ be arbitrarily small. By definition, for all sufficiently large $\ell$,
\begin{align*}
T_i(\ell) \leq q^{(c_i+\delta)\ell}, \qquad \text{for $i=1,2$.}
\end{align*}
Hence, there exist real constants $\gamma_1,\gamma_2$ such that for all $\ell$,
\begin{align*}
T_i(\ell) \leq \gamma_i q^{(c_i+\delta)\ell}, \qquad \text{for $i=1,2$.}
\end{align*}
Plugging this into~\eqref{eq:prod}, and using $c_1 \geq c_2$, we obtain
\[
\abs*{\cA_{\cI}(n)}=\sum_{i=0}^{n} T_1(i)T_2(n-i) \leq n\gamma_1\gamma_2 q^{(c_1+\delta)n}.
\]
Using this in the definition for capacity we therefore have
\[
\ccap(\cI) \leq c_1+\delta,
\]
and since this holds for all $\delta>0$, necessarily
\[
\ccap(\cI) \leq c_1.
\]
For the reverse inequality we note that
\[
\abs*{\cA_{\cI}(n)} = \sum_{i=0}^{n} T_1(i)T_2(n-i) \geq T_1(n),
\]
implying
\[
\ccap(\cI) \geq c_1,
\]
which completes the proof.
\end{IEEEproof}

It is known~\cite[Lemma 1]{BarWacYaa25} that if $\cI=(I_1)$, i.e., there is a single coloring channel, then $\ccap(\cI)=\log_q\abs{I}$. Additionally, if $\cI=(I_1,\dots,I_t)$, $\abs{I_1}=\dots=\abs{I_t}$, and $I_i\cap I_j=\emptyset$ for all $i\neq j$, then it was proved in~\cite[Theorem 3]{BarWacYaa25} that $\ccap(\cI)=\log_q\abs{I_1}$. This type of result is immediately extended by the following corollary to channels of arbitrary size:

\begin{corollary}
Let $\cI=(I_1,\dots,I_t)$ be a sequence of coloring channels over $\Sigma=[q]$. Assume $I_i\cap I_j=\emptyset$ for all $i\neq j$. Then
\[
\ccap(\cI) = \log_q \max\set*{\abs*{I_1},\dots,\abs*{I_t}}.
\]
\end{corollary}
\begin{IEEEproof}
By~\cite[Lemma 1]{BarWacYaa25}, the capacity of a single channel is $\ccap((I_i))=\log_q\abs{I_i}$. We then use Lemma~\ref{lem:trivial}.
\end{IEEEproof}

In what follows, in order to avoid trivial cases we define the following type of set systems.

\begin{definition}
Let $\cI=(I_1,\dots,I_t)$ be a sequence of coloring channels over $\Sigma$. We say $\cI$ is \emph{irreducible} if $I_i\not\subseteq I_j$ for all $i\neq j$, and $\cI$ is not separable.
\end{definition}

\subsection{Uniform sunflowers}

The first irreducible set family for which we compute the exact capacity is a uniform sunflower, which is defined as follows:

\begin{definition}
A set family $\cI = \set{I_1,\ldots,I_t}\subseteq 2^{[q]}$ is called a \emph{$(k,p,t)$-sunflower} if all the following conditions hold:
\begin{enumerate}
\item 
$\abs{\bigcap_{i\in [t]} I_i} = k$,
\item
$\abs{I_i} = k+p$, for all $i\in [t]$,
\item
$I_u\cap I_v = \bigcap_{i\in [t]}I_i$, for all $u\neq v \in [t]$.
\end{enumerate}
\end{definition}

\begin{theorem}\label{thm:sunflower}
Fix an alphabet $\Sigma=[q]$, and a sequence of coloring channels $\cI = (I_1,\dots,I_t)$, $I_i\subseteq\Sigma$ for all $i$. If $\cI$ is a $(k,p,t)$-sunflower, $k,p,t\geq 1$, then
\[\ccap(\cI) = g(y^*),\]
where 
\[g(y) = (1-y)\log_q k + y\log_q p + (t-(t-1)y)H\parenv*{\frac{y}{t-(t-1)y}}\log_q 2,\]
and where $y^*$ is the unique root of
\[\frac{pt (1-y)^t}{ky (1-(t-1)y/t)^{t-1}} = 1,\]
in $(0,1)$.
\end{theorem}
\begin{IEEEproof}
We partition the alphabet $\Sigma$ into $t+2$ parts, 
\begin{align*}
K &\eqdef \bigcap_{\ell \in [t]} I_\ell, &
I_\ell^* &\eqdef I_\ell \setminus K, &
L &\eqdef \Sigma \setminus \bigcup_{\ell \in [t]} I_\ell.
\end{align*}
For integers $i_1, \dots, i_t, j_1, j_2 \geq 0$ with $j_1 + j_2 + \sum_{\ell \in [t]} i_\ell = n$, we define $\cA^{i_1,\ldots,i_t,j_1,j_2}$ as the set of all sequences $\bx \in \Sigma^n$ such that $\bx$ contains exactly $i_\ell$ entries from $I_\ell^*$ for each $\ell \in [t]$, $j_1$ entries from $K$, and $j_2$ entries from $L$. We then define
\[
\cA_{\cI}^{i_1,\ldots,i_t,j_1,j_2} \eqdef \set*{\bx_{\cI} : \bx \in \cA^{i_1,\ldots,i_t,j_1,j_2}}.
\]

Fix some $\bx_{\cI} = (\bx_{I_1}, \ldots, \bx_{I_t}) \in  \cA_{\cI}^{i_1,\ldots,i_t,j_1,j_2}$. For each $\ell\in[t]$,  
$\bx_{I_\ell}$ is an interleaving of $\bx_K$ and $\bx_{I_\ell^*}$. The number of ways to combine $\bx_{I_\ell^*}$ with $\bx_K$ in $\cA_{\cI}^{i_1,\ldots,i_t,j_1,j_2}$ is $\binom{j_1 + i_\ell}{i_\ell}$, which corresponds to the number of ways to insert $i_\ell$ entries into a sequence of length $j_1$. Hence, we have
\[
\abs*{\cA_{\cI}^{i_1,\ldots,i_t,j_1,j_2}} = k^{j_1} p^{\sum_{\ell \in [t]} i_\ell} \prod_{\ell \in [t]} \binom{j_1 + i_\ell}{i_\ell}.
\]

Since there are at most $n^{t+2}$ choices of $(i_1,\dots,i_t,j_1,j_2)$, and
\[
\cA_{\cI} = \bigcup_{\substack{i_1,\dots,i_t,j_1,j_2 \geq 0 \\ j_1 + j_2 + \sum_{\ell \in [t]} i_\ell = n}} \cA_{\cI}^{i_1,\ldots,i_t,j_1,j_2},
\]  
we obtain
\[
\max \abs*{\cA_{\cI}^{i_1,\ldots,i_t,j_1,j_2}} \leq \abs*{\cA_{\cI}} \leq n^{t+2} \max \abs*{\cA_{\cI}^{i_1,\ldots,i_t,j_1,j_2}}.
\]
Since $\frac{(t+2) \log_q n}{n} \to 0$ as $n\to\infty$, the desired capacity is
\[
\ccap(\cI) = \limsup_{n \to \infty} \frac{1}{n}\log_q {\abs*{\cA_{\cI}}}  = \limsup_{n \to \infty} \frac{1}{n}\log_q \max \abs*{\cA_{\cI}^{i_1,\ldots,i_t,j_1,j_2}}.
\]

Define 
\begin{align*}
\alpha_{-1} &\eqdef \frac{j_2}{n}, & \alpha_0 &\eqdef \frac{j_1}{n}, & \alpha_\ell &\eqdef \frac{i_\ell}{n},
\end{align*}
for all $\ell\in[t]$, so that $\alpha_{-1} + \alpha_0 + \sum_{\ell \in [t]} \alpha_\ell = 1$. Then, by using~\eqref{eq:binom},
\[
\frac{1}{n} \log_q \abs*{\cA_{\cI}^{i_1,\ldots,i_t,j_1,j_2}} 
= \alpha_0 \log_q k + \sum_{\ell \in [t]} \alpha_\ell \log_q p + \sum_{\ell \in [t]} (\alpha_0 + \alpha_\ell) H\parenv*{\frac{\alpha_\ell}{\alpha_0 + \alpha_\ell}} \log_q 2 + o(1).
\]

For fixed $\alpha_{-1}, \alpha_0 \in [0,1]$, define
\[
M(x_1,\ldots,x_t) \eqdef \alpha_0 \log_q k + \sum_{\ell \in [t]} x_\ell \log_q p + \sum_{\ell \in [t]} (\alpha_0 + x_\ell) H\parenv*{\frac{x_\ell}{\alpha_0 + x_\ell}} \log_q 2,
\] 
where $x_\ell \in [0,1]$. Then,
\[
\ccap(\cI) = \max_{\sum_{\ell \in [t]} x_\ell = 1 - \alpha_0 - \alpha_{-1}} M(x_1,\ldots,x_t) = \max_{\sum_{\ell \in [t]} x_\ell = 1 - \alpha_0} M(x_1,\ldots,x_t),
\]
where the last equality holds because $M$ does not depend on $\alpha_{-1}$, so the maximum is trivially achieved at $\alpha_{-1} = 0$. If we define $\overline{x}=\sum_{i\in[t]}x_t /t$ to be the average of $x_1,\dots,x_t$, then by concavity
\begin{align*}
M(x_1,\ldots,x_t) &= \alpha_0 \log_q k + (1-\alpha_0) \log_q p + \sum_{\ell \in [t]} (\alpha_0 + x_\ell) H\parenv*{\frac{x_\ell}{\alpha_0 + x_\ell}} \log_q 2 \\
& \leq \alpha_0 \log_q k + (1-\alpha_0) \log_q p + t (\alpha_0 + \overline{x}) H\parenv*{\frac{\overline{x}}{\alpha_0 + \overline{x}}} \log_q 2 \\
&= M(\overline{x},\dots,\overline{x}).
\end{align*}
Thus, the maximum is attained at $x_1 = \dots = x_t \eqdef x$. Then, we simply need to find
\[
\max_{0 \leq x \leq 1/t} (1 - t x) \log_q k + t x \log_q p + t \big(1 - (t-1)x \big) H\parenv*{ \frac{x}{1 - (t-1)x} } \log_q 2.
\]
By substituting $y = t x$, we obtain that the desired maximum is $\max_{0 \leq y \leq 1} g(y)$, where 
\[
g(y) \eqdef (1-y) \log_q k + y \log_q p + \parenv*{t - (t-1) y} H\parenv*{ \frac{y}{t - (t-1)y} } \log_q 2.
\]
By direct calculation, we have the first and second derivatives of $g(y)$,
\begin{align*}
    g'(y) &= \log_q \parenv*{ \frac{p t (1-y)^t}{k y (1-(t-1)y/t)^{t-1}} },\\
    g''(y) &= - \frac{t}{\ln q \, y (1-y) (t-(t-1)y)} < 0.
\end{align*}
Moreover, $\lim_{y \to 0^+} g'(y) = +\infty$ and $\lim_{y \to 1^-} g'(y) = -\infty$, so there exists a unique root $y^* \in (0,1)$ such that $g'(y^*) = 0$. Thus, the desired maximum is $g(y^*)$.
\end{IEEEproof}

Theorem~\ref{thm:sunflower} contains~\cite[Th.~2]{BarWacYaa25} as a special case, as the latter is simply a $(q-2,1,2)$-sunflower. It significantly extends the cases for which we can calculate the capacity precisely. 

\subsection{Two intersecting sets}

If in the previous section we discussed a very uniform structure (same sized sets with a strict intersection configuration), here we relax these conditions but focus on two sets only.

\begin{theorem}\label{thm:twosets}
Fix an alphabet $\Sigma=[q]$, and a sequence of two coloring channels $\cI = (I_1,I_2)$ with $I_1,I_2 \subseteq \Sigma_q$. Furthermore, assume
\begin{align*}
\abs*{I_1\cap I_2} &= k, & \abs*{I_1\setminus I_2} &= p_1, & \abs*{I_2\setminus I_1} &= p_2,
\end{align*}
for some integers $k,p_1,p_2\geq 1$. Then
\[\ccap(\cI) = M(x^*_1,x^*_2),\]
where
\begin{align*}
M(x_1,x_2) &\eqdef (1-x_1-x_2) \log_q k +  x_1 \log_q p_1 + x_2 \log_q p_2 \\
&\qquad + (1-x_2) H\parenv*{\frac{x_1}{1-x_2}} \log_q 2 + (1-x_1) H\parenv*{\frac{x_2}{1-x_1}} \log_q 2,
\end{align*}
and
\begin{equation}
\label{eq:x1x2}
\begin{split}
x^*_1 &= \frac{1}{2}-\frac{k+p_2-p_1}{2\sqrt{(k+p_1+p_2)^2-4p_1p_2}}, \\
x^*_2 &= \frac{1}{2}-\frac{k+p_1-p_2}{2\sqrt{(k+p_1+p_2)^2-4p_1p_2}}.
\end{split}
\end{equation}
\end{theorem}

\begin{IEEEproof}
Define
\begin{align*}
K& \eqdef I_1\cap I_2, &
I^*_1 &\eqdef I_1\setminus I_2, &
I^*_2 &\eqdef I_2\setminus I_1, &
L& \eqdef[q]\setminus (I_1\cup I_2).
\end{align*}
Then $\abs{K} = k$, $\abs{I^*_1} = p_1$, $\abs{I^*_2} = p_2$, and $\abs{L}=q-k-p_1-p_2$. For integers $i_1, i_2, j_1, j_2 \geq 0$ with $j_1 + j_2 + i_1+i_2 = n$, we define $\cA^{i_1,i_2,j_1,j_2}$ as the set of all sequences $\bx \in \Sigma^n$ such that $\bx$ contains exactly $i_\ell$ entries from $I_\ell^*$ for each $\ell \in [2]$, $j_1$ entries from $K$, and $j_2$ entries from $L$. We then define
\[
\cA_{\cI}^{i_1,i_2,j_1,j_2} \eqdef \set*{\bx_{\cI} : \bx \in \cA^{i_1,i_2,j_1,j_2}}.
\]
By the same method as the proof of Theorem~\ref{thm:sunflower}, we have 
\[
\abs*{\cA_{\cI}^{i_1,i_2,j_1,j_2}} = k^{j_1} p_1^{i_1}p_2^{i_2} \binom{j_1 + i_1}{i_1}\binom{j_1 + i_2}{i_2},
\]
and the capacity is
\[
\ccap(\cI) = \limsup_{n \to \infty} \frac{1}{n}\log_q \max \abs*{\cA_{\cI}^{i_1,i_2,j_1,j_2}}.
\]

 Define 
\begin{align*}
\alpha_{-1} &\eqdef \frac{j_2}{n}, & \alpha_0 &\eqdef \frac{j_1}{n}, & \alpha_1 &\eqdef \frac{i_1}{n}, & \alpha_2 &\eqdef \frac{i_2}{n},
\end{align*}
so that $\alpha_{-1} + \alpha_0 + \alpha_1 +\alpha_2 = 1$. Then
\begin{align*}
\frac{1}{n} \log_q \abs*{\cA_{\cI}^{i_1,i_2,j_1,j_2}} 
&= \alpha_0 \log_q k +  \alpha_1 \log_q p_1 + \alpha_2 \log_q p_2\\
& \qquad + (\alpha_0 + \alpha_1) H\parenv*{\frac{\alpha_1}{\alpha_0 + \alpha_1}} \log_q 2 + (\alpha_0 + \alpha_2) H\parenv*{\frac{\alpha_2}{\alpha_0 + \alpha_2}} \log_q 2 + o(1).
\end{align*}
Like in the proof of Theorem~\ref{thm:sunflower}, it is clear that the capacity is maximized when $\alpha_{-1}=0$, i.e., $\alpha_0=1-\alpha_1-\alpha_2$. Then,
\[
\ccap(\cI) = \max_{\substack{x_1+x_2\leq 1 \\ x_1,x_2\geq 0}} M(x_1,x_2).
\]

To solve this optimization problem we take the following steps. First, we compute the Hessian matrix:
\[
\bH= \begin{pmatrix}
\frac{\partial^2 M}{\partial x_1^2} & \frac{\partial^2 M}{\partial x_1 \partial x_2} \\
\frac{\partial^2 M}{\partial x_2 \partial x_1} & \frac{\partial^2 M}{\partial x_2^2} \\
\end{pmatrix}
=
\begin{pmatrix}
\frac{x_1+x_2-2x_1 x_2 -1}{x_1(1-x_1)(1-x_1-x_2)} &
-\frac{2}{1-x_1-x_2} \\
-\frac{2}{1-x_1-x_2} & \frac{x_1+x_2-2x_1 x_2 -1}{x_2(1-x_2)(1-x_1-x_2)}
\end{pmatrix}
\log_q 2.
\]
One can verify that for any $x_1,x_2>0$, $x_1+x_2<1$, and any $\bv\in\R^2$, $v\neq 0$, we have $\bv\bH \bv^\intercal<0$. Thus, $\bH$ is negative definite, and $M(x_1,x_2)$ is strictly concave. It follows that the function is maximized when $\nabla M(x_1,x_2)=0$. We note that
\begin{align*}
\frac{\partial M}{\partial x_1} & = \log_q\parenv*{\frac{(1-x_1-x_2)^2}{x_1(1-x_1)}}-\log_q k + \log_q p_1, \\
\frac{\partial M}{\partial x_2} & = \log_q\parenv*{\frac{(1-x_1-x_2)^2}{x_2(1-x_2)}}-\log_q k + \log_q p_2.
\end{align*}
Equating these two to $0$, the maximum of $M(x_1,x_2)$, $x_1,x_2\geq 0$, $x_1+x_2\leq 1$, is obtained at $(x^*_1,x^*_2)$ as in~\eqref{eq:x1x2}.
\end{IEEEproof}

\subsection{Paths}

The final irreducible set family we study is a path, which is defined as follows:

\begin{definition}
A set family $\cI = \set{I_1,\ldots,I_t}\subseteq 2^{[q]}$ is called a \emph{path of length $t$} if for all $i\in[t]$, $I_i=\set{\sigma_{i-1},\sigma_i}$, where $\sigma_0,\dots,\sigma_t\in\Sigma$ are distinct letters.
\end{definition}

While seemingly simple, finding the exact capacity of paths involves an elaborate optimization problem. As we shall soon show, this problem relies heavily on Chebyshev polynomials and their properties. We shall therefore review the relevant known facts about Chebyshev polynomials. The reader is referred to~\cite{MasHan03} for an extensive study of these polynomials.

Chebyshev polynomials have four variants. We shall require those of the second and fourth kind. The Chebyshev polynomial of the second kind is denoted by $U_i(x)$, and that of the fourth kind by $W_i(x)$. They are the unique polynomials satisfying
\begin{align}
U_i(\cos\theta) &= \frac{\sin((i+1)\theta)}{\sin \theta}, &
W_i(\cos\theta) &= \frac{\sin((i+\frac{1}{2})\theta)}{\sin (\frac{1}{2}\theta)}, \label{eq:UWsin}
\end{align}
(see~\cite[Eq.~1.4 and Eq.~1.9]{MasHan03}). It is known~\cite[Eq.~1.51]{MasHan03} that for all $i\geq 0$,
\begin{equation}
\label{eq:chebstandard}
U_i\parenv*{\frac{x+x^{-1}}{2}}=\frac{x^{i+1}-x^{-(i+1)}}{x-x^{-1}}.
\end{equation}
The roots of $U_i(x)$ are $-1 < \nu_1 < \dots < \nu_k < 1$,
\begin{equation}
\label{eq:Uroot}
\nu_k = \cos \parenv*{\frac{(i-k+1)\pi}{i+1}}, \qquad \text{for $k\in[i]$,}
\end{equation}
and those of $W_i(x)$ are $-1 < \omega_1 < \dots < \omega_k < 1$,
\begin{equation}
\label{eq:Wroot}
\omega_k = \cos\parenv*{\frac{(i-k+1)\pi}{i+\frac{1}{2}}}, \qquad\text{for $k\in[i]$,}
\end{equation}
(see~\cite[Eq.~2.4 and Eq.~2.10]{MasHan03}). Additionally~\cite[Eq.~1.18]{MasHan03},
\begin{equation}
\label{eq:WsumU}
W_i(x) = U_i(x) + U_{i-1}(x),
\end{equation}
as well as~\cite[Eq.~2.29a and Eq.~2.29b]{MasHan03}
\begin{align}
\label{eq:UW1}
U_i(1) &= i+1, & W_i(1) &= 2i+1.
\end{align}

We now turn to prove a few technical lemmas that will be instrumental in proving the main capacity result for paths.

\begin{lemma}
\label{lem:rroot}
Let $m\in(0,4)$ be a real number and recursively define
\[
r_i = \frac{(m-1)r_{i-1}-1}{r_{i-1}+1},
\]
for all $i\geq 1$, with $r_0=m-1$. Then
\[r_i = \frac{\lambda_1^{i+2} - \lambda_2^{i+2}}{\lambda_1^{i+1} - \lambda_2^{i+1}} -1,\]
where
\begin{align*}
\lambda_1 &= \frac{m+\sqrt{m^2-4m}}{2}, \\
\lambda_2 &=\frac{m-\sqrt{m^2-4m}}{2},
\end{align*}
are the (possibly complex) roots of $x^2-mx+m=0$.
\end{lemma}

\begin{IEEEproof}
Define $s_i\eqdef r_i+1$, for all $i\geq 0$. Then writing the definition of $r_i$ in terms of $s_i$, we have $s_0=m$, and
\[
s_i = m-\frac{m}{s_{i-1}}.
\]
Hence, we have a truncated continued fraction
\[
s_i = m-\cfrac{m}{m-\cfrac{m}{m-\cfrac{m}{\ddots-\cfrac{m}{m}}}},
\]
where there are a total of $i$ fraction lines. To solve this truncated continued fraction, by~\cite[p.~15, Eq.~(1.3) and~(1.4)]{Wal48}, there exist two linear recurrences, $A_i$ and $B_i$, such
that
\begin{align*}
A_{-1}&=1, & A_0 &= m, & A_{i+1} &= m A_{i} - mA_{i-1},\\
B_{-1}&=0, & B_0 &= 1, & B_{i+1} &= m B_{i} - mB_{i-1},\\
\end{align*}
and
\[
s_i = \frac{A_i}{B_i}.
\]
By definition, $B_1 = m$, and both $A_i$ and $B_i$ satisfy the same linear recurrence, so we have $A_i=B_{i+1}$, for all $i\geq 0$.

Solving this linear recurrence is straightforward. It has a characteristic equation
\[
x^2 - mx + m = 0,
\]
with roots $\lambda_1$ and $\lambda_2$. Thus,
\[
B_i = c_1 \lambda_1^i + c_2 \lambda_2^i,
\]
with appropriately chosen constants $c_1$ and $c_2$. Using the base cases for the recursion, we can solve and get
\begin{align*}
c_1 &= \frac{\lambda_1}{\sqrt{m^2-4m}}, & c_2 &= -\frac{\lambda_2}{\sqrt{m^2-4m}}.
\end{align*}
Hence, using $A_i=B_{i+1}$, we obtain
\[
s_i = \frac{c_1 \lambda_1^{i+1} + c_2 \lambda_2^{i+1}}{c_1 \lambda_1^i + c_2 \lambda_2^i} = \frac{\lambda_1^{i+2}-\lambda_2^{i+2}}{\lambda_1^{i+1}-\lambda_2^{i+1}}.
\]
Finally, recalling that $r_i=s_i-1$, the proof is complete.
\end{IEEEproof}

Using Lemma~\ref{lem:rroot}, we can now show a connection to Chebyshev polynomials.

\begin{lemma}
\label{lem:rcheb}
Let $r_i$, $m$, and $\lambda_1,\lambda_2$, be defined as in Lemma~\ref{lem:rroot}, and define $u\eqdef \frac{m-2}{2}$. Then, for all
$i\geq 1$,
\begin{align*}
r_{2i-1} &= \frac{U_i(u)}{U_{i-1}(u)}, &
r_{2i} &= \frac{W_{i+1}(u)}{W_{i}(u)},
\end{align*}
where $U_\ell(x)$ and $W_\ell(x)$ are the Chebyshev polynomials of the second and fourth kind, respectively (see~\eqref{eq:UWsin}).
\end{lemma}

\begin{IEEEproof}
Define $s\eqdef \lambda_1-1$. Since $m\neq 0$, we have $s\neq 0$. We observe that $\lambda_2-1 = s^{-1}$, and that 
\begin{equation}
\label{eq:ssum}
s+s^{-1}=m-2=2u.
\end{equation}
Define $z\eqdef \sqrt{s}$. Then
\begin{align*}
1+s & = z\parenv*{z+z^{-1}}, &
1+s^{-1} &= z^{-1}\parenv*{z+z^{-1}}.
\end{align*}
With these, it now follows that for all $i\geq 1$,
\begin{equation}
\label{eq:lampowdif}
\lambda_1^i-\lambda_2^i = (1+s)^i-\parenv*{1+s^{-1}}^i
= \parenv*{z+z^{-1}}^i\parenv*{z^i-z^{-i}}.
\end{equation}
Define 
\[A_i\eqdef z^i-z^{-i}.\]
Then
\begin{equation}
\label{eq:Aspread}
\parenv*{z+z^{-1}}A_i = A_{i+1}+A_{i-1}.
\end{equation}
We therefore have
\begin{align}
r_i &\overset{(a)}{=} \frac{\lambda_1^{i+2}-\lambda_2^{i+2}}{\lambda_1^{i+1}-\lambda_2^{i+1}}-1 \overset{(b)}{=} \frac{(z+z^{-1})^{i+2}A_{i+2}}{(z+z^{-1})^{i+1}A_{i+1}}-1
=\frac{(z+z^{-1})A_{i+2}}{A_{i+1}}-1 \nonumber \\
&\overset{(c)}{=} \frac{A_{i+3}+A_{i+1}}{A_{i+1}}-1 = \frac{A_{i+3}}{A_{i+1}}, \label{eq:riratio1}
\end{align}
where $(a)$ follows from Lemma~\ref{lem:rroot}, $(b)$ follows from~\eqref{eq:lampowdif}, and $(c)$ follows from~\eqref{eq:Aspread}.

We can now prove the first half of the claim:
\[
r_{2i-1} \overset{(a)}{=} \frac{A_{2i+2}}{A_{2i}} = \frac{z^{2i+2}-z^{-(2i+2)}}{z^{2i}-z^{-2i}} \overset{(b)}{=} \frac{s^{i+1}-s^{-(i+1)}}{s^i-s^{-i}}
\overset{(c)}{=}\frac{U_i\parenv*{\frac{s+s^{-1}}{2}}}{U_{i-1}\parenv*{\frac{s+s^{-1}}{2}}}=\frac{U_i(u)}{U_{i-1}(u)},
\]
where $(a)$ follows from~\eqref{eq:riratio1}, $(b)$ follows from $z=\sqrt{s}$, $(c)$ follows from~\eqref{eq:chebstandard}, and $(d)$ follows from~\eqref{eq:ssum}. The second half is similar,
\begin{align*}
r_{2i} &= \frac{A_{2i+3}}{A_{2i+1}} 
= \frac{(z+z^{-1})A_{2i+3}}{(z+z^{-1})A_{2i+1}}
= \frac{A_{2i+4}+A_{2i+2}}{A_{2i+2}+A_{2i}}\\
&= \frac{z^{2i+4}-z^{-(2i+4)}+z^{2i+2}-z^{-(2i+2)}}{z^{2i+2}-z^{-(2i+2)}+z^{2i}-z^{-2i}} \\
&=\frac{s^{i+2}-s^{-(i+2)}+s^{i+1}-s^{-(i+1)}}{s^{i+1}-s^{-(i+1)}+s^i-s^{-i}} \\
&= \frac{U_{i+1}\parenv*{\frac{s+s^{-1}}{2}}+U_i\parenv*{\frac{s+s^{-1}}{2}}}{U_i\parenv*{\frac{s+s^{-1}}{2}}+U_{i-1}\parenv*{\frac{s+s^{-1}}{2}}}=\frac{U_{i+1}(u)+U_i(u)}{U_i(u)+U_{i-1}(u)}\\
&= \frac{W_{i+1}(u)}{W_i(u)},
\end{align*}
where the last equality follows from~\eqref{eq:WsumU}.
\end{IEEEproof}

The final useful lemma is the following:

\begin{lemma}
\label{lem:chebsol}
Let $U_i(x)$ and $W_i(x)$ be Chebyshev polynomials of the second and fourth kind, respectively. Then:
\begin{enumerate}
    \item
    The solutions of
    \[
    \frac{U_i\parenv*{\frac{x-2}{2}}}{U_{i-1}\parenv*{\frac{x-2}{2}}} = \frac{1}{x-1}
    \]
    in the interval $(0,4)$ are
    \[
    x_k = 2+2\cos\parenv*{\frac{2\pi k}{2i+3}},
    \]
    for $k\in[i+1]$, and $k\neq\frac{2i+3}{3}$ if $3|i$.
    \item
    The solutions of
    \[
    \frac{W_i\parenv*{\frac{x-2}{2}}}{W_{i-1}\parenv*{\frac{x-2}{2}}} = \frac{1}{x-1}
    \]
    in the interval $(0,4)$ are
    \[
    x_k = 2+2\cos\parenv*{\frac{2\pi k}{2i+2}},
    \]
    for $k\in[i]$, and $k\neq\frac{2(i+1)}{3}$ if $3|(i+1)$.
\end{enumerate}
\end{lemma}
\begin{IEEEproof}
If $x\in(0,4)$, then $\frac{x-2}{2}\in(-1,1)$. We start with the first claim. Set $\frac{x-2}{2}=\cos\theta$. By~\eqref{eq:UWsin}, the equation we are trying to solve becomes
\[
\frac{\sin((i+1)\theta)}{\sin(i\theta)}=\frac{1}{2\cos\theta+1}.
\]
Denote $j\eqdef\sqrt{-1}$ and $z\eqdef e^{j\theta}$. Then $\cos(\ell\theta)=\frac{1}{2}(z^\ell+z^{-\ell})$ while $\sin(\ell\theta)=\frac{1}{2j}(z^\ell-z^{-\ell})$. The above equation then becomes
\[
\frac{z^{i+1}-z^{-(i+1)}}{z^i-z^{-i}} = \frac{1}{z+z^{-1}+1}.
\]
After rearranging we get
\[
(z+1)(z^{2i+3}-1)=0.
\]
We note that $z=-1$ is not a solution to the original equation. So we are left with solving $z^{2i+3}=1$, which gives candidate solutions
\[
\theta = \frac{2\pi k}{2i+3},
\]
for any integer $0\leq k \leq 2i+2$. We rule out $k=0$ since this causes a division by $0$ on the LHS, as well as $k=\frac{2i+3}{3},\frac{2(2i+3)}{3}$ (if those are integers) since these cause a division by $0$ on the RHS. We also eliminate duplicates, since $\cos\theta = \cos (-\theta)$. Returning to the original variable $x$ we have the desired solution.

The second claim proceeds along the same lines. By~\eqref{eq:UWsin} our equation becomes
\[
\frac{\sin\parenv*{\parenv*{i+\frac{1}{2}}\theta}}{\sin\parenv*{\parenv*{i-\frac{1}{2}}\theta}}=\frac{1}{2\cos\theta+1}.
\]
Rewriting this in $z$ and rearranging we get
\[
(z+1)(z^{2i+2}-1)=0.
\]
Once again $z=-1$ is not a solution. Solving $z^{2i+2}=1$ gives candidate solutions
\[
\theta=\frac{2\pi k}{2i+2},
\]
with $0\leq k\leq 2i+1$. The solutions where $k=0,i+1$ are discarded, as well as $k=\frac{2(i+1)}{3},\frac{4(i+1)}{3}$, if integers. Duplicates are removed, thus giving us the desired claim.
\end{IEEEproof}

We are now in a position to state and prove the capacity of a path of length $t$.

\begin{theorem}
\label{thm:path}
Fix an alphabet $\Sigma=[q]$. Let $\cI = (I_1,\dots,I_t)$ be path of length $t\geq 3$. Let
\[
m^* = 2+2\cos\parenv*{\frac{2\pi}{t+3}},
\]
and let $r^*_i$, $i\geq 0$ be given recursively by $r^*_0=m^*-1$, and for all $i\geq 1$,
\[
r^*_i = \frac{(m^*-1)r^*_{i-1}-1}{r^*_{i-1}+1}.
\]
For $0\leq i\leq t$, let
\[
\alpha^*_i = \frac{\prod_{j=0}^{i-1}r^*_j}{\sum_{\ell=0}^t\prod_{j=0}^{\ell-1}r^*_j}.
\]
Then
\[\ccap(\cI) = \sum_{i\in [t]}(\alpha^*_{i-1}+\alpha^*_i) H\parenv*{\frac{\alpha^*_{i}}{\alpha^*_{i-1} + \alpha^*_{i}}} \log_q 2.\]
\end{theorem}

\begin{IEEEproof}
Let $I_i=\set{\sigma_{i-1},\sigma_i}$ for all $i\in [t]$, and where $\sigma_0,\dots,\sigma_t\in\Sigma$ are distinct letters. We partition the alphabet $\Sigma$ into $t+2$ parts: the singletons
$\set{\sigma_0},\dots, \set{\sigma_t}$, and the letters not on the path, $L \eqdef \Sigma \setminus \set{\sigma_0,\dots,\sigma_t}$.
For non-negative integers $a_0,a_1 \dots, a_t, b\geq 0$, with $b+ \sum_{i \in [0,t]} a_i  = n$, we define $\cA^{a_0,a_1\ldots,a_t,b}$ as the set of all sequences $\bx \in \Sigma^n$ such that $\bx$ contains exactly $a_i$ entries of $\sigma_i$, for each $i \in [0,t]$, and $b$ entries from $L$. We then define
\[
\cA_{\cI}^{a_0,a_1,\ldots,a_t,b} \eqdef \set*{\bx_{\cI} : \bx \in \cA^{a_0,a_1,\ldots,a_t,b}}.
\]
The number of ways of choosing $a_0,a_1,\ldots,a_t,b$ is upper bounded by $(n+1)^{t+1}$ (each of the $a_i$ is an integer chosen between $0$ and $n$, and $b$ completes the sum, if possible, to $n$), and therefore
\[
\max_{a_0,\dots,a_t,b}\abs*{\cA_\cI^{a_0,\dots,a_t,b}} \leq \abs*{\cA_\cI} \leq \sum_{a_0,\dots,a_t,b}\abs*{\cA_\cI^{a_0,\dots,a_t,b}} \leq (n+1)^t \max_{a_0,\dots,a_t,b}\abs*{\cA_\cI^{a_0,\dots,a_t,b}}.
\]
Since $\frac{1}{n}\log_q (n+1)^{t+1} \to 0$ as $n\to\infty$, the desired capacity is 
\[
\ccap(\cI) = \limsup_{n \to \infty} \frac{1}{n}\log_q \max_{a_0,\dots,a_t,b} \abs*{\cA_{\cI}^{a_0,a_1,\ldots,a_t,b}}.
\]

Next, we find $\abs{\cA_{\cI}^{a_0,a_1,\ldots,a_t,b}}$. Assume $\bx_{\cI} = (\bx_{I_1}, \ldots, \bx_{I_t}) \in \cA_{\cI}^{a_0,a_1,\ldots,a_t,b}$. For each $i\in[t]$, $\bx_{I_i}$ is a sequence with $a_{i-1}$ occurrences of the letter $v_{i-1}$ and $a_i$ occurrences of $v_i$. Hence, there are at most $\binom{a_{i-1}+a_i}{a_i}$ possible such sequences. In total, we have
\[
\abs*{\cA_{\cI}^{a_0,a_1,\ldots,a_t,b}} \leq \prod_{i=1}^t \binom{a_{i-1}+a_i}{a_i}.
\]
We contend this counting is exact, i.e., any possible sequence is attainable. To see this, assume an arbitrary choice of $(\bx_1,\dots,\bx_t)$, where for all $i\in[t]$, $\bx_i$ contains exactly $a_{i-1}$ occurrences of $\sigma_{i-1}$, exactly $a_i$ occurrences of $\sigma_i$, and nothing else. We show that there exists a sequence $\bx\in\Sigma^n$ such that $\bx_\cI=(\bx_1,\dots,\bx_t)$. Construct $\bx$ iteratively as follows. Start with the sequence of length $a_0$ containing only the letter $\sigma_0$. Then insert $a_1$ copies of $\sigma_1$ so that the sequence agrees with $\bx_1$. Next, do the same with $a_2$ copies of $\sigma_2$, so there is agreement with $\bx_2$. Continue the process until agreeing with $\bx_t$. Finally, insert arbitrary $b$ letters from $L$ in arbitrary positions. By construction, $\bx_\cI=(\bx_1,\dots,\bx_t)$. It follows that
\begin{equation}
\label{eq:absAI}
\abs*{\cA_{\cI}^{a_0,a_1,\ldots,a_t,b}} = \prod_{i=1}^t \binom{a_{i-1}+a_i}{a_i}.
\end{equation}

When maximizing $\abs{\cA_\cI^{a_0,\dots,a_t,b}}$, we note that
\[
\abs*{\cA_\cI^{a_0,\dots,a_{t-1},a_t,b}} \leq \abs*{\cA_\cI^{a_0,\dots,a_{t-1},a_t+b,0}}.
\]
Thus, the maximum is attained when $b=0$. Let us now define $\alpha_i\eqdef \frac{a_i}{n}$ for all $i\in [0,t]$, so $\sum_{i\in[0,t]}\alpha_i=1$. Let
\[
G(\alpha_0,\ldots,\alpha_t) = \sum_{i\in [t]}(\alpha_{i-1}+\alpha_i) H\parenv*{\frac{\alpha_{i}}{\alpha_{i-1} + \alpha_{i}}} \log_q 2.
\]
Then by~\eqref{eq:absAI} and~\eqref{eq:binom},
\[
    \frac{1}{n} \log_q \abs*{\cA_{\cI}^{a_1,\ldots,a_t,0}} = G(\alpha_0,\ldots,\alpha_t) + o(1),
\]
which implies that
\[
\ccap(\cI) = \max_{\substack{\alpha_0,\dots,\alpha_t\geq 0\\ \alpha_0+\dots+\alpha_t=1}} G(\alpha_0,\dots,\alpha_t).
\]

We contend that $G(\alpha_0,\dots,\alpha_t)$ is concave in our domain. To show that, we examine the Hessian, $\bH$, where $\bH_{i,j}=\frac{\partial^2 G}{\partial \alpha_i \partial \alpha_j}$. By calculation,
 \begin{align*}
    \frac{\partial^2 G}{\partial \alpha_0^2} &= \frac{-\alpha_1}{\alpha_0(\alpha_0+\alpha_1)\ln q}, \\
    \frac{\partial^2 G}{\partial \alpha_i^2} &= \frac{-\alpha_{i-1}}{\alpha_{i}(\alpha_{i-1}+\alpha_i)\ln q} + \frac{-\alpha_{i+1}}{\alpha_{i}(\alpha_{i+1}+\alpha_i)\ln q},\quad \text{ for } 1\leq i \leq t-1, \\
    \frac{\partial^2 G}{\partial \alpha_t^2} &= \frac{-\alpha_{t-1}}{\alpha_{t}(\alpha_{t-1}+\alpha_t)\ln q}, \\
    \frac{\partial^2 G}{\partial \alpha_i\partial \alpha_{i+1}} &= \frac{\partial^2 G}{\partial \alpha_{i+1}\partial \alpha_i}= \frac{1}{(\alpha_i+\alpha_{i+1})\ln q},~\text{ for } 0\leq i \leq t-1,\\
    \frac{\partial^2 G}{\partial \alpha_i\partial \alpha_{j}} &= 0,~\text{ for } \abs{i-j}\geq 2.
 \end{align*}
For any vector $\bv=(v_0,v_1,\ldots,v_t)\in \R^{t+1}$, $\bv\neq 0$, and any $(\alpha_0,\dots,\alpha_t)\in (0,1)^{t+1}$, we have
\begin{align*}
    \bv \bH \bv^\intercal &= \sum_{i=0}^t \frac{\partial^2 G}{\partial \alpha_i^2}v_i^2 + 2\sum_{i=0}^{t-1} \frac{\partial^2 G}{\partial \alpha_i\partial \alpha_{i+1}} v_i v_{i+1}\\
    &=\sum_{i=0}^{t-1} \frac{-1}{(\alpha_i+\alpha_{i+1})\ln q} \parenv*{\sqrt{\frac{\alpha_{i+1}}{\alpha_i}}v_i - \sqrt{\frac{\alpha_{i}}{\alpha_{i+1}}}v_{i+1}}^2\\
    &\leq 0.
\end{align*}
Thus, $G(\alpha_0,\dots,\alpha_t)$ is concave over $[0,1]^{t+1}$, and also with the added constraint of $\alpha_0+\dots+\alpha_t=1$. It follows that it suffices to find a local maximum in this domain, which will also give us the global maximum.

To maximize $G(\alpha_0,\dots,\alpha_t)$ under the constraints $\alpha_0,\dots,\alpha_t\geq 0$, $\alpha_0+\dots+\alpha_t=1$, we employ the Lagrange multiplier method. Define
\begin{align*}
\cL(\alpha_0,\dots,\alpha_t,\lambda) & \eqdef G(\alpha_0,\dots,\alpha_t)+\lambda\cdot g(\alpha_0,\dots,\alpha_t), \\
g(\alpha_0,\dots,\alpha_t) &\eqdef \alpha_0+\dots+\alpha_t-1.
\end{align*}
We now need to solve $\nabla \cL = 0$. The first $t+1$ derivatives, with respect to $\alpha_0,\dots,\alpha_t$, give us the following equations:
\begin{equation}
\label{eq:mainder}
\begin{split}
0 &= \frac{\partial \cL}{\partial \alpha_0} = \log_q \frac{\alpha_0+\alpha_1}{\alpha_0} +\lambda,\\
0 &= \frac{\partial \cL}{\partial \alpha_i} = \log_q \frac{\alpha_{i-1}+\alpha_i}{\alpha_i} + \log_q \frac{\alpha_{i+1}+\alpha_i}{\alpha_i}+\lambda,\quad \text{ for } 1\leq i \leq t-1,\\
0&=\frac{\partial \cL}{\partial \alpha_t} = \log_q \frac{\alpha_{t-1}+\alpha_t}{\alpha_t}+\lambda,
\end{split}
\end{equation}
and the last derivative, with respect to $\lambda$, gives us the
constraint back,
\begin{equation}
\label{eq:const}
0 = \frac{\partial \cL}{\partial \lambda} = \alpha_0+\dots+\alpha_t-1.
\end{equation}

Define $m\eqdef q^{-\lambda}$, and $r_i \eqdef \frac{\alpha_{i+1}}{\alpha_i}$, for all $i\in [0,t-1]$. We then have~\eqref{eq:mainder} become
\begin{align*}
    m &= 1 + r_0,\\ 
    m &= (1+1/r_{i-1})(1+r_i), \quad \text{ for } 1\leq i \leq t-1,\\
    m &= 1 + \frac{1}{r_{t-1}}. 
\end{align*}
Since all the $\alpha_i$ are non-negative, so must be the $r_i$, which we shall later guarantee. Rearranging for $r_i$, we obtain
\begin{align}
r_0 &= m-1, \label{eq:recurbase} \\ 
r_i &= \frac{(m-1)r_{i-1}-1}{r_{i-1}+1}, \quad \text{ for } 1\leq i \leq t-1, \label{eq:recurrule} \\
r_{t-1} &= \frac{1}{m-1}.  \label{eq:recurend}
\end{align}
We shall attempt to find solutions involving $m\in(0,4)$ (since by previous arguments, it suffices to find a local maximum, which by concavity attains the global maximum). Using~\eqref{eq:recurbase} and~\eqref{eq:recurrule} with Lemma~\ref{lem:rcheb}, the recursion is solved giving an explicit formula for the $r_i$. In particular, for $r_{t-1}$ we obtain
\[
r_{t-1} = \begin{cases}
\frac{U_{t/2}(\frac{m-2}{2})}{U_{(t-2)/2}(\frac{m-2}{2})} & \text{$t$ is even,} \\
\frac{W_{(t+1)/2}(\frac{m-2}{2})}{W_{(t-1)/2}(\frac{m-2}{2})} & \text{$t$ is odd.}
\end{cases}
\]

Now that we have found $r_{t-1}$ we can use~\eqref{eq:recurend} in order to find $m$. If $t$ is even, by~\eqref{eq:recurend} we need to solve
\[
\frac{U_{t/2}(\frac{m-2}{2})}{U_{(t-2)/2}(\frac{m-2}{2})} = \frac{1}{m-1}.
\]
By Lemma~\ref{lem:chebsol} there are several solutions. We need to pick one that guarantees all of $r_0,r_1,\dots,r_{t-1}$ are non-negative. We contend the largest root satisfies this, so we
choose
\[
m=2+2\cos\parenv*{\frac{2\pi}{t+3}}.
\]
We observe that indeed $m\in (0,4)$, and $r_0 > 0$ since $m>2$. Additionally, $r_1,\dots,r_{t-1}$ use ratios of $U_0,\dots,U_{t/2}$ and $W_1,\dots,W_{t/2}$. All of these have $U_\ell(1),W_\ell(1)>0$ (see~\eqref{eq:UW1}), and the rightmost root happens to be of $U_{t/2}$, which is $\cos(\frac{2\pi}{t+2})$. But the $m$ we chose gives $\frac{m-2}{2}=\cos(\frac{2\pi}{t+3})> \cos(\frac{2\pi}{t+2})$, and so all of $U_\ell(\frac{m-2}{2}),W_\ell(\frac{m-2}{2})$ are positive, thus $r_0,\dots,r_{t-1}>0$.

A similar arguments holds for $t$ odd. By~\eqref{eq:recurend} we solve
\[
\frac{W_{(t+1)/2}(\frac{m-2}{2})}{W_{(t-1)/2}(\frac{m-2}{2})} = \frac{1}{m-1}.
\]
We use Lemma~\ref{lem:chebsol} and again choose
\[
m = 2+2\cos\parenv*{\frac{2\pi}{t+3}}.
\]
As before, $m\in(0,4)$ and $r_0>0$. Our $r_1,\dots,r_{t-1}$ use ratios of $U_0,\dots,U_{(t-1)/2}$ and $W_1,\dots,W_{(t+1)/2}$. Their largest root is that of $W_{(t+1)/2}$, which is $\cos(\frac{2\pi}{t+2})$, so again, $r_0,\dots,r_{t-1}>0$ for the same reasons as in the even case.

To conclude, we set the desired sequence $\alpha_0,\dots,\alpha_t$ to have ratios $r_0,\dots,r_{t-1}$, and normalize them to satisfy~\eqref{eq:const}, so
\[
\alpha_i = \frac{\prod_{j=0}^{i-1}r_j}{\sum_{\ell=0}^t\prod_{j=0}^{\ell-1}r_j},
\]
which completes our proof.
\end{IEEEproof}

\section{Bounds on the Capacity}
\label{sec:bound}

In this section our goal is to provide bounds on the capacity, in cases where an exact capacity is not known. We start by showing a reduction of any sequence of coloring channels to another sequence of coloring channels with the exact same capacity, but with channels containing only two letters. This provides an intuitive explanation to a result from~\cite{BarWacYaa25}, and allows us to prove a general bound. Then, with an eye on the special case of $q=4$, we note that even after all the results of Section~\ref{sec:exact}, there is a single missing case for which we do not know the capacity, and that is for a sequence of coloring channels that form a cycle. We prove a specialized bound for this system.

We start with the reduction approach, and the following definition:

\begin{definition}
Given a set system $\cI\subseteq 2^\Sigma$, we define the \emph{pairs graph of $\cI$} as $P_\cI=(\Sigma,E_\cI)$, with edges
\[
E_\cI \eqdef \set*{\set{u,v} : u,v \in I, u\neq v, I \in \cI}.
\]  
\end{definition}

By imposing an arbitrary ordering on the set of edges, $E_\cI$, we can treat it as a sequence of $\abs{E_\cI}$ coloring channels, each defined by an edge, thus containing exactly two letters.

\begin{example}
Assume $q=4$, and let $\cI=\set{I_1,I_2}$, with $I_1=\set{1,2,3}$ and $I_2=\set{2,3,4}$. Then the pairs graph, $P_\cI$ has vertices $\Sigma=[q]=\set{1,2,3,4}$, and edges $E_{\cI}=\set{\set{1,2},\set{1,3},\set{2,3},\set{2,4},\set{3,4}}$. If we arbitrarily order $E_\cI$ we obtain a sequence of five coloring channels, each with two letters.
\end{example}

We shall relate the capacity of $E_\cI$ to that of $\cI$ in the following theorem. Since the case of a single coloring channel is already solved, we shall focus on two or more coloring channels.

\begin{theorem}\label{thm:pairs graph}
Fix an alphabet $\Sigma=[q]$, and an irreducible sequence of coloring channels $\cI=(I_1,\dots,I_t)$, $t\geq 2$, $I_i\subseteq \Sigma$ for all $i$. Then for all $n$ we have 
\[
\abs*{\cA_{\cI}}=\abs*{\cA_{E_\cI}},
\]
and therefore also,
\[
\ccap(\cI) = \ccap(E_{\cI}).
\]
\end{theorem}

\begin{IEEEproof}
Denote $E_{\cI} = (J_1, \dots, J_m)$. Observe that since $t\geq 2$ and $\cI$ is irreducible, necessarily $\abs{I_i}\geq 2$ for all $i\in[t]$.

In the first direction, for any $\bx, \by \in \Sigma^n$ with $\bx_{E_{\cI}} \neq \by_{E_{\cI}}$, there exists some $\ell \in [m]$ such that $\bx_{J_\ell} \neq \by_{J_\ell}$. By definition, $J_\ell \subseteq I$ for some $I \in \cI$. Then clearly $\bx_I \neq \by_I$, which implies $\bx_{\cI} \neq \by_{\cI}$. It follows that any reconstruction code for $E_{\cI}$ is a reconstruction code for $\cI$, and hence, $\abs{\cA_{E_{\cI}}} \leq \abs{\cA_{\cI}}$.

In the other direction, assume $\cC$ is a reconstruction code for $\cI$. We will show that it is also a reconstruction code for $E_\cI$. Our strategy will be the following. We will prove that for any codeword $\bx\in\cC$, we can use $\bx_{E_\cI}$ to deduce $\bx_\cI$, and then find $\bx$, thus showing $\cC$ can indeed reconstruct its codewords under the channels $\bx_{E_\cI}$.

Fix some $\bx\in\cC$, and some $I\in\cI$. Denote the edges contributed to $E_\cI$ by $I$ as
\[
\binom{I}{2}=(K_1,\dots,K_s).
\]
We shall show how to find $\bx_I$ from $\bx_{K_1},\dots,\bx_{K_s}$. As observed at the beginning, $\abs{I}\geq 2$, and so $s\geq 1$. Assume $\bx_I = (z_1, \dots, z_{n'})$. We know the composition of $\bx_I$ by counting the number of occurrences of each symbol in $\bx_{K_1},\dots,\bx_{K_s}$. We then discover the identity of $z_1$ by the following process. If $a,b\in\Sigma$, assume $K_i=\set{a,b}$. Whichever of $a$ and $b$ that \emph{does not} appear first in $\bx_{K_i}$, is ruled out of being $z_1$. Since we can do this for every pair of possible symbols, the identity of $z_1$ is uniquely discovered. We then discard the first symbol of $\bx_I$, and the first occurrence of $z_1$ from any $\bx_{K_i}$, and repeat the process to find $z_2$, and so on.

Having found $\bx_I$ from $\bx_{K_1},\dots,\bx_{K_s}$, we can repeat the process for any $I\in\cI$. Thus, from $\bx_{E_\cI}$ we can find $\bx_I$. Since $\cC$ is a reconstruction code for $\cI$, we can find $\bx$. Thus, any reconstruction code for $\cI$ is also a reconstruction code for $E_{\cI}$, and so $\abs{\cA_\cI}\leq \abs{\cA_{E_\cI}}$.

Combining the two inequalities, we infer that $\abs{\cA_\cI}=\abs{\cA_{E_{\cI}}}$, and by definition, $\ccap(\cI)=\ccap(E_{\cI})$.
\end{IEEEproof}

We note that Theorem~\ref{thm:pairs graph} generalizes one direction of~\cite[Theorem 4]{BarWacYaa25}. The latter showed that $\abs{\cA_{\cI}}=q^n$ if and only if every pair of letters is contained in some coloring channel. In Theorem~\ref{thm:pairs graph} this is equivalent to the pairs graph being a $q$-clique, hence, the sequence of coloring channels $\cI$ is equivalent (in terms of optimal code size and capacity) to a single coloring channel that contains all the letters, $I=[q]$, which trivially allows a reconstruction code of length $n$ and size $q^n$.

Also, as a result of simple monotonicity, we obtain the following straightforward corollary:

\begin{corollary}
\label{cor:subgraph}
Let $\cI$ and $\cI'$ be two irreducible sequences of coloring channels over $\Sigma$. If $P_{\cI}$ is a subgraph of $P_{\cI'}$, then
\[ \abs*{\cA_{\cI}} \leq \abs*{\cA_{\cI'}},\]
and
\[
\ccap(\cI) \leq \ccap(\cI').
\]
\end{corollary}

Another remark we make is that a sequence of coloring channels, $\cI$, that induces a pairs graph $P_\cI$, may have a smaller equivalent sequence of coloring channels, $\cI'$ with the same graph. To find that, we need to find the smallest edge-clique cover of $P_\cI$. Each clique in this cover becomes a channel in $\cI'$. The number of channels required in $\cI'$ is then the \emph{intersection number} of $P_\cI$ (see~\cite{ErdGooPos66}).

Finding general bounds on $\ccap(\cI)$ for arbitrary $\cI$, is a difficult problem. We present a crude lower and upper bound that are based on the pairs graph, $P_\cI$.

\begin{theorem}\label{thm:general}
Fix an alphabet $\Sigma=[q]$, and an irreducible sequence of coloring channels $\cI=(I_1,\dots,I_t)$, $t\geq 2$, $I_i\subseteq\Sigma$ for all $i\in[t]$. Let $P_\cI$ be the pairs graph of $\cI$, and let $k\eqdef \omega(P_\cI)$ denote the size of the largest clique in $P_\cI$. Then
\[\log_q k \leq \ccap(\cI) \leq \log_q (kte).\]
\end{theorem}

\begin{IEEEproof}
Let $K$ be a largest clique in the graph $P_{\cI}$, with size $\abs{K}=\omega(P_{\cI}) \eqdef k$. For the lower bound, we use Corollary~\ref{cor:subgraph} with $K$ being a subgraph of $P_{\cI}$. Since the capacity of a clique is $\log_q k$, the lower bound is proved.

We turn to prove the upper bound. By abuse of notation, let $K$ denote the vertices of the clique, i.e., $K\subseteq\Sigma$. Since every coloring channel $I_i$ creates a clique in the pairs graph $P_\cI$, by definition we have that $\abs{I_i}\leq k$ for all $i\in[t]$. We define the following sequence of coloring channels,
$\cI' = \set{I_1', \dots, I_t'}$, where $I_i' = K \cup I_i$ for all $i\in[t]$ (and if necessary, removing coloring channels that are contained in another coloring channel, so that $\cI'$ is irreducible). Since $P_{\cI} \subseteq P_{\cI'}$, we have $\ccap(\cI) \leq \ccap(\cI')$ by Corollary~\ref{cor:subgraph}. Thus, we continue by finding an upper bound on $\ccap(\cI')$.

Define 
\begin{align*}
I_0' & \eqdef K, & I'_{\leq \ell} & \eqdef \bigcup_{j\in[\ell]} I_j',\\
Q_i & \eqdef I_i'\cap I_{\leq i-1}', & R_i & \eqdef I_i' \setminus Q_i,
\end{align*}
for all $i\in [t]$. It is clear that $K \subseteq Q_i$, and that
\[\abs*{R_i} \leq  \abs*{I_i \setminus K} \leq \max \set*{\abs*{I_j}: j\in [t]} \leq k.\]
Additionally, $K,R_1,\dots,R_t$ are pairwise disjoint. Define $L \eqdef \Sigma \setminus (K\cup R_1\cup\dots\cup R_t)$.

For integers $i_1, \dots, i_t, j_1, j_2 \geq 0$ with $j_1+ j_2 + \sum_{\ell \in [t]} i_\ell  = n$, we define $\cA^{i_1,\ldots,i_t,j_1,j_2}$ as the set of all sequences $\bx \in \Sigma^n$ such that $\bx$ contains exactly $i_\ell$ entries from $R_\ell$ for each $\ell \in [t]$, $j_1$ entries from $K$, and $j_2$ entries from $L$. We then define
\[
\cA_{\cI'}^{i_1,\ldots,i_t,j_1,j_2} \eqdef \set*{\bx_{\cI'} : \bx \in \cA^{i_1,\ldots,i_t,j_1,j_2}}.
\]
Since there are at most $(n+1)^{t+2}$ ways of choosing $i_1,\dots,i_t,j_1,j_2$, we have
\[
\abs*{\cA_{\cI'}} \leq (n+1)^{t+2} \max_{i_1,\dots,i_t,j_1,j_2}\abs*{\cA_{\cI'}^{i_1,\ldots,i_t,j_1,j_2}}.
\]
By the definition of capacity,
\[
\ccap(\cI') \leq \limsup_{n \to \infty} \frac{1}{n}\log_q \max_{i_1,\dots,i_t,j_1,j_2} \abs*{\cA_{\cI'}^{i_1,\ldots,i_t,j_1,j_2}}.
\]

To upper bound $\abs{\cA_{\cI'}^{i_1,\dots,i_t,j_1,j_2}}$, we consider the following iterative process. Assume $\bx_{\cI'} = (\bx_{I_1'}, \ldots, \bx_{I_t'}) \in  \cA_{\cI'}^{i_1,\ldots,i_t,j_1,j_2}$. For the upper bound with first need to choose which $j_1$ letters from $K$ appear, and in which order, for a total of $k^{j_1}$ options. Then, looking at channel $I'_1$ we need the identity of $i_1$ letters, for which we have $\abs{R_1}^{i_1}$ options. However, these need to be placed among the previously $j_1$ letters, in at most $\binom{j_1+i_1}{i_1}$ ways. For channel $I'_2$ we have $\abs{R_2}^{i_2}$ choice of letters, which need to be placed among the previously chosen letters, $Q_2$. Since $\abs{Q_2}\leq j_1+i_1$, there are at most $\binom{j_1+i_1+i_2}{i_2}$ ways of placing those letters in the output of channel $I'_2$. Continuing along these lines, we obtain
\begin{align*}
\abs*{\cA_{\cI'}^{i_1,\ldots,i_t,j_1,j_2}} &\leq  k^{j_1} \prod_{\ell \in [t]} \abs*{R_i}^{i_\ell} \prod_{\ell\in [t]} \binom{j_1+\sum_{s=1}^{\ell} i_s}{i_\ell}\\
&\leq k^{j_1+\sum_{\ell \in [t]} i_\ell} \prod_{\ell\in [t]} \binom{j_1+\sum_{s=1}^{\ell} i_s}{i_\ell},
\end{align*}
where we used the fact that $\abs{R_i}\leq \abs{I_i}\leq k$.

Define 
\begin{align*}
\alpha_{-1} &\eqdef \frac{j_2}{n}, & \alpha_0 &\eqdef \frac{j_1}{n}, & \alpha_\ell &\eqdef \frac{i_\ell}{n},
\end{align*}
for all $\ell\in[t]$, so that $\alpha_{-1} + \alpha_0 + \sum_{\ell \in [t]} \alpha_\ell = 1$. Then, together with~\eqref{eq:binom},
\begin{align*}
\frac{1}{n} \log_q \abs*{\cA_{\cI'}^{i_1,\dots,i_t,j_1,j_2}}\leq &\sum_{\ell=0}^t \alpha_\ell \log_q k + \sum_{\ell \in [t]} \parenv*{\sum_{s=0}^\ell \alpha_s} H\parenv*{\frac{\alpha_\ell}{\sum_{s=0}^\ell \alpha_s}} \log_q 2 + o(1).
\end{align*}
Obviously, the maximum is attained at $\alpha_{-1}=0$. We then have 
\begin{align*}
\ccap(\cI) & \leq \ccap(\cI') \\
&\leq \log_q k +\max_{\substack{\alpha_0,\ldots,\alpha_t \geq 0 \\ \sum_{\ell=0}^t \alpha_\ell = 1}}  \sum_{\ell \in [t]} \parenv*{\sum_{s=0}^\ell \alpha_s} H\parenv*{\frac{\alpha_\ell}{\sum_{s=0}^\ell \alpha_s}} \log_q 2\\
&\overset{(a)}{\leq} \log_q k + \max_{\substack{\alpha_0,\ldots,\alpha_t \geq 0\\ \sum_{\ell=0}^t \alpha_\ell = 1}} \sum_{\ell \in [t]}  H(\alpha_\ell) \log_q 2 \\
& \overset{(b)}{\leq} \log_q k + \parenv*{t H(1/t) + \log_2 e}\log_q 2 \\
& \overset{(c)}{\leq} \log_q (kte),
\end{align*}
where $(a)$ follows from concavity of $H$ so $xH(y/x)\leq H(y)$, for $x\geq y$, $(b)$ follows again from concavity of $H$ implying the optimization yields $\alpha_1=\dots=\alpha_t$, and $(c)$ follows from the fact that $tH(1/t)-\log_2 t$ is a strictly increasing function with a limit of $\log_2 e$ as $t\to\infty$.
\end{IEEEproof}

The bound above is rather weak. We now study the only missing case in the catalog of coloring channels over $q=4$, which is a cycle.

\begin{definition}
\label{def:cycle}
A set family $\cI = \set{I_1,\ldots,I_t}\subseteq 2^{[q]}$ is called a \emph{cycle of length $t$} if for all $i\in[t]$, $I_i=\set{\sigma_{i-1},\sigma_i}$ (with indices taken cyclically, i.e., $I_1=\set{\sigma_t,\sigma_1}$), where $\sigma_1,\dots,\sigma_t\in\Sigma$ are distinct letters.
\end{definition}

We note that cycles of length $1$ or $2$ are degenerate, and a cycle of length $3$ has a pairs graph that is a clique. Thus, we focus on the unsolved cases of cycles of length $t\geq 4$.

\begin{theorem}
\label{thm:cycle}
Fix an alphabet $\Sigma=[q]$. Let $\cI=(I_1,\dots,I_t)$ be a cycle of length $t\geq 4$. Then
\[
c_{t-1} \leq \ccap(\cI) \leq \begin{cases}
\parenv*{\frac{1}{\sqrt{3}}+\parenv*{1+\frac{1}{\sqrt{3}}}H(2-\sqrt{3})}\log_q 2 \approx \log_q 3.732 & t=4 \\
\log_q 4 & t \geq 5,
\end{cases}
\]
where $c_{t-1}$ is the capacity of a path of length $t-1$ as given in Theorem~\ref{thm:path}.
\end{theorem}

\begin{IEEEproof}
For the lower bound we note that removing one channel, say, $I_t$, results in a path of length $t-1$, whose pairs graph is a subgraph of the pairs graph of the cycle. Thus, the lower bound follows by Corollary~\ref{cor:subgraph}.

For the upper bound, let $I_i=\set{\sigma_{i-1},\sigma_i}$, as in Definition~\ref{def:cycle}. When $t=4$, the pairs graph, $P_\cI$, is simply a cycle of length $4$, containing the edges $\set{\sigma_1,\sigma_2},\set{\sigma_2,\sigma_3},\set{\sigma_3,\sigma_4},\set{\sigma_4,\sigma_1}$. Define $\cI'=(I'_1,I'_2)$, with $I'_1=\set{\sigma_1,\sigma_2,\sigma_3}$, and $I'_2=\set{\sigma_3,\sigma_4,\sigma_1}$. We note that $P_\cI$ is a subgraph of $P_{\cI'}$, and so $\ccap(\cI)\leq \ccap(\cI')$ by Corollary~\ref{cor:subgraph}. But $\cI'$ is a $(2,1,2)$-sunflower. The upper bound in this case is given by Theorem~\ref{thm:sunflower}.

The final case is the upper bound for $t\geq 5$. We follow a similar logic as in the proof of Theorem~\ref{thm:path}. Let $\cA^{a_1,\dots,a_t,b}$ denote the set of sequences of length $n$ over $\Sigma$, with $a_i$ occurrences of $\sigma_i$, for all $i\in[t]$, and $b$ occurrences of letters from $\Sigma\setminus\set{\sigma_1,\dots,\sigma_t}$. Define
\[
\cA^{a_1,\dots,a_t,b}_\cI \eqdef \set*{ \bx_{\cI} : \bx\in\cA^{a_1,\dots,a_t,b}}.
\]
As in the proof of Theorem~\ref{thm:path},
\[
\ccap{\cI} = \limsup_{n\to\infty}\frac{1}{n}\log_q \max_{a_1,\dots,a_t,b}\abs*{\cA^{a_1,\dots,a_t,b}_\cI}.
\]
If we observe a general $\bx_\cI=(\bx_{I_1},\dots,\bx_{I_t})\in \cA^{a_1,\dots,a_t,b}_\cI$, then there are at most $\binom{a_{i-1}+a_i}{a_i}$ possible sequences for $\bx_{I_i}$, so
\[
\abs*{\cA^{a_1,\dots,a_t,b}_\cI} \leq \prod_{i\in[t]}\binom{a_{i-1}+a_i}{a_i},
\]
where indices are taken cyclically. Additionally, the maximal size is obviously obtained when $b=0$. Writing $\alpha_i\eqdef \frac{a_i}{n}$ and using~\eqref{eq:binom}, we therefore have
\begin{align*}
\ccap(\cI) &= \limsup_{n\to\infty}\frac{1}{n}\log_q \max_{a_1,\dots,a_t,b}\abs*{\cA^{a_1,\dots,a_t,b}_\cI} \\
& \leq \max_{\substack{\alpha_1,\dots,\alpha_t\geq 0 \\ \alpha_1+\dots+\alpha_t=1}} \sum_{i\in[t]} (\alpha_{i-1}+\alpha_i) H\parenv*{\frac{\alpha_i}{\alpha_{i-1}+\alpha_i}}\log_q 2 \\
& \leq \max_{\substack{\alpha_1,\dots,\alpha_t\geq 0 \\ \alpha_1+\dots+\alpha_t=1}} \sum_{i\in[t]} (\alpha_{i-1}+\alpha_i) \log_q 2 \\
& = 2\log_q 2 = \log_q 4.
\end{align*}
\end{IEEEproof}

\section{Conclusion}
\label{sec:conc}

In this paper we studied the problem of determining the capacity of the coloring channel. Previous results~\cite{BarWacYaa25}, managed to find the exact capacity of a single channel, two channels that form a $(q-2,1,2)$-sunflower, or equal-sized channels that form disjoint sets. We generalized the latter in Lemma~\ref{lem:trivial} to any separable sequence of coloring channels. We also generalized the former in Theorem~\ref{thm:sunflower} to any $(k,p,t)$-sunflower. We also added exact capacities for two arbitrary intersecting sets in Theorem~\ref{thm:twosets}, as well as paths in Theorem~\ref{thm:path}. We showed that the capacity in fact depends entirely on the pairs graph, which we used to give bounds on the capacity of general coloring channels. We concluded by giving a bound specifically for pairs graphs that form a cycle.

In light of the pairs-graph approach, when the alphabet is ternary, $q=3$, there only two irreducible sequences of coloring channels that use all three letters. The capacities of these two can already be deduced from the results of~\cite{BarWacYaa25}, as shown in Table~\ref{tab:q3}. However, for $q=4$ there are six cases, only two of which covered by~\cite{BarWacYaa25}. As a consequence of our results, we can now give the exact capacity of all irreducible coloring channels over an alphabet of size $q=4$, except for the case of a cycle, in which we only have bounds. This catalog of capacities is given in Table~\ref{tab:q4}, where degenerate cases are omitted. The omitted cases include separable coloring channels, or channels that do not use all of the alphabet letters. These may be reduced to smaller alphabets, and are easily solvable using the tools given in this paper.

Several open question remain. First, when looking at Table~\ref{tab:q4} and Theorem~\ref{thm:cycle}, we do not yet have a closed form solution for the capacity of a cycle. More generally, the exact capacities obtained in this paper are for pairs graphs all of whose cycles are contained in cliques. Solving these kinds sequences of coloring channels seems a challenging combinatorial optimization problem.

Second, while this paper finds the exact capacity of several sequences of coloring channels, we still lack a nice description of the reconstruction codes attaining the capacity asymptotically.

Finally, if such codes as above are found, an important component is missing for applying these codes in practice: we need to find efficient encoding and reconstruction procedures. The former translate arbitrary user messages into codewords, while the latter see the channel outputs and reconstruct the original transmitted codeword.

\begin{table*}
\caption{The capacity of all irreducible coloring channels over an alphabet of size $3$, that use all possible letters, to $5$ significant digits}
\label{tab:q3}
\begin{center}
\renewcommand{\arraystretch}{1.2}
\begin{tabular}{cclll}
\hline\hline
$P_\cI$ & $\ccap(\cI)$ & Example minimal $\cI$ & Type & Location \\
\hline\hline
\includegraphics[width=1.5em]{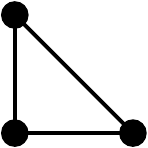} & 1 & $(\set{1,2,3})$ & clique & \cite[Lemma 1]{BarWacYaa25} \\
\includegraphics[width=1.5em]{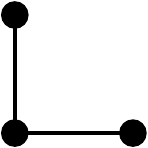} & $0.87604$ & $(\set{1,3},\set{2,3})$ & $(1,1,2)$-sunflower or two sets & \cite[Th.~2]{BarWacYaa25} or Theorem~\ref{thm:sunflower} or Theorem~\ref{thm:twosets}\\
\hline\hline
\end{tabular}
\end{center}
\end{table*}

\begin{table*}
\caption{The capacity of all irreducible coloring channels over an alphabet of size $4$, that use all possible letters, to $5$ significant digits}
\label{tab:q4}
\begin{center}
\renewcommand{\arraystretch}{1.2}
\begin{tabular}{cclll}
\hline\hline
$P_\cI$ & $\ccap(\cI)$ & Example minimal $\cI$ & Type & Location \\
\hline\hline
\includegraphics[width=1.5em]{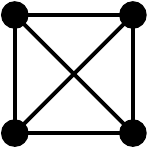} & 1 & $(\set{1,2,3,4})$ & clique & \cite[Lemma 1]{BarWacYaa25} \\
\includegraphics[width=1.5em]{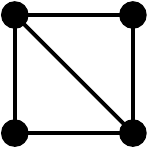} & $0.94998$ & $(\set{1,2,3},\set{1,3,4})$ & $(2,1,2)$-sunflower or two sets & \cite[Th.~2]{BarWacYaa25} or Theorem~\ref{thm:sunflower} or Theorem~\ref{thm:twosets}\\
\includegraphics[width=1.5em]{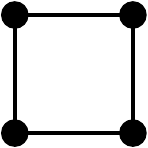} & $\in[0.79248,0.94998]$ & $(\set{1,2},\set{2,3},\set{3,4},\set{4,1})$ & cycle of length $4$ & Theorem~\ref{thm:cycle}\\
\includegraphics[width=1.5em]{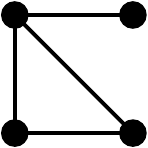} & $0.88578$ & $(\set{1,2},\set{1,3,4})$ & two sets & Theorem~\ref{thm:twosets}\\
\includegraphics[width=1.5em]{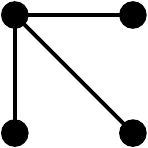} & $0.82720$ & $(\set{1,2},\set{1,3},\set{1,4})$ & $(1,1,3)$-sunflower & Theorem~\ref{thm:sunflower}\\
\includegraphics[width=1.5em]{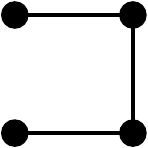} & $0.79248$ & $(\set{1,2},\set{2,3},\set{3,4})$ & path of length $3$ & Theorem~\ref{thm:path}\\
\hline\hline
\end{tabular}
\end{center}
\end{table*}

\bibliographystyle{IEEEtranS}
\bibliography{allbib}

\end{document}